\documentclass{siamart1116}

\usepackage{amsfonts,amssymb,amsmath,latexsym,ae,aecompl,bbm,algorithm}
\usepackage{enumitem}
\usepackage{multirow}
\usepackage[noend]{algpseudocode}
\usepackage{graphicx}
\usepackage{color}

\newcommand{\red}[1]{{\textcolor{red}{#1}}}
\newcommand{\blue}[1]{\textcolor{blue}{#1}}
\newcommand{\magenta}[1]{\textcolor{magenta}{#1}}

\newcommand{\gd}[1]{\ensuremath{\mathsf{Greedy}[#1]}}
\newcommand{\lt}[1]{\ensuremath{\mathsf{Left}[#1]}}
\newcommand{\fd}[1]{\ensuremath{\mathsf{FirstDiff}[#1]}}
\newcommand{\alg}[1]{\ensuremath{\mathsf{Alg}[#1]}}
\newcommand{\prob}{\mathbf{Pr}}

\algblockdefx[ExecBl]{BlockOn}{BlockOff} [1]{#1}

\makeatletter
\ifthenelse{\equal{\ALG@noend}{t}}
{\algtext*{BlockOff}}
{}
\makeatother

\def\ShowComment{True}
\ifdefined\ShowComment
\def\billy#1{\marginpar{$\leftarrow$\fbox{B}}\footnote{$\Rightarrow$~{\sf #1 \blue{--Billy}}}}
\def\amanda#1{\marginpar{$\leftarrow$\fbox{A}}\footnote{$\Rightarrow$~{\sf #1 \red{--Amanda}}}}
\def\john#1{\marginpar{$\leftarrow$\fbox{J}}\footnote{$\Rightarrow$~{\sf #1 \magenta{--John}}}}
\else
\def\billy#1{}
\def\amanda#1{}
\def\john#1{}
\fi

\usepackage{epstopdf}
\ifpdf
  \DeclareGraphicsExtensions{.eps,.pdf,.png,.jpg}
\else
  \DeclareGraphicsExtensions{.eps}
\fi

\numberwithin{theorem}{section}

\newcommand{\TheTitle}{Balanced Allocation: Patience is not a Virtue} 
\newcommand{\TheAuthors}{J. Augustine, W. K. Moses Jr., A. Redlich, and E. Upfal}

\headers{\TheTitle}{\TheAuthors}

\title{{\TheTitle}\thanks{An earlier version of this paper appeared in \cite{AMRU16}. The previous version had only expectation upper bounds for average number of probes and no lower bounds. This version has lower bounds on maximum load and high probability upper bounds for average number of probes, as well as cleaner proofs for the same.
}}

\author{
  John Augustine\thanks{Department of Computer Science \& Engineering, Indian Institute of Technology Madras, Chennai, India.
    (\email{augustine@iitm.ac.in}). Supported by the IIT Madras New Faculty Seed Grant, the IIT Madras Exploratory Research Project, and the Indo-German Max Planck Center for Computer Science (IMPECS).}
  \and
  William K. Moses Jr.\thanks{Department of Computer Science \& Engineering, Indian Institute of Technology Madras, Chennai, India.
      (\email{wkmjr3@gmail.com}).}
  \and
  Amanda Redlich\thanks{Department of Mathematics, Bowdoin College, ME, USA (\email{aredlich@bowdoin.edu}). This material is based upon work supported by the National Science Foundation under Grant No. DMS-0931908 while the author was in residence at the Institute for Computational and Experimental Research in Mathematics in Providence, RI, during the Spring 2014 semester.}
  \and
    Eli Upfal\thanks{Department of Computer Science, Brown University, RI, USA (\email{eli@cs.brown.edu}).}
}

\usepackage{amsopn}

\ifpdf
\hypersetup{
  pdftitle={\TheTitle},
  pdfauthor={\TheAuthors}
}
\fi

\theoremstyle{empty}
\newtheorem{duplicate}{NameIgnored}

\begin{document}

\date{}

\maketitle

\begin{abstract} 
 
Load balancing is a well-studied problem, with balls-in-bins being the primary framework. The greedy algorithm $\mathsf{Greedy}[d]$ of Azar et al.  places each ball by probing $d > 1$ random bins and placing the ball in the least loaded of them. With high probability, the maximum load under $\mathsf{Greedy}[d]$ is exponentially lower than the result when balls are placed uniformly randomly.
V\"ocking showed that a slightly asymmetric variant, $\mathsf{Left}[d]$, provides a further significant improvement.  However, this improvement comes at an additional computational cost of imposing structure on the bins.  

Here, we present a fully decentralized and easy-to-implement algorithm called $\mathsf{FirstDiff}[d]$ that combines the simplicity of  $\mathsf{Greedy}[d]$ and the improved balance of $\mathsf{Left}[d]$.   The key idea in $\mathsf{FirstDiff}[d]$ is to probe until a different bin size from the first observation is located, then place the ball. Although  the number of probes could be quite large for some of the balls,
we show that $\mathsf{FirstDiff}[d]$ requires only at most $d$ probes on average per ball (in both the standard and the heavily-loaded settings).  Thus the number of probes is no greater than either that of  $\mathsf{Greedy}[d]$ or $\mathsf{Left}[d]$. 
More importantly, we show that $\mathsf{FirstDiff}[d]$ closely matches the improved maximum load ensured by $\mathsf{Left}[d]$ in both the standard and heavily-loaded settings. We further provide a tight lower bound on the maximum load up to $O(\log \log \log n)$ terms.
We additionally give experimental data that $\mathsf{FirstDiff}[d]$ is indeed as good as $\mathsf{Left}[d]$, if not better, in practice.

%
\end{abstract}

\begin{keywords}
Load balancing,
FirstDiff,
Balanced allocation,
Randomized algorithms,
Task allocation
\end{keywords}

\begin{AMS}
  60C05, 60J10, 68R05
\end{AMS}


\section{Introduction}
\label{sec:intro}

Load balancing is the study of distributing loads across multiple entities such that the load is minimized across all the entities.  This problem arises naturally in many settings, including the distribution of requests across multiple servers, in peer-to-peer networks when requests need to be spread out amongst the participating nodes, and in hashing.   Much research has focused on practical implementations of solutions to these problems \cite{SX07,FXS11,SLZGJCP14}.  

Our work builds on several classic algorithms in the theoretical balls-in-bins model.  In this model, $m$ balls are to be placed sequentially into $n$ bins and each ball probes the load in random bins in order to make its choice.  Here we give a new algorithm, \fd{d}, which performs as well as the best known algorithm, \lt{d}, while being significantly easier to implement.

The allocation time for a ball is the number of probes made to different bins before placement. The challenge is to balance the allocation time versus the maximum bin load.  For example, using one probe per ball,  i.e. placing each ball uniformly at random, the maximum load of any bin when $m=n$ will be $\frac{\ln n}{\ln \ln n} (1 + o(1))$ (with high probability\footnote{We use the phrase ``with high probability" (or w.h.p. in short) to denote probability of the form $1 - O(n^{-c})$ for some suitable $c > 0$. Furthermore, every $\log$ in this paper is to base 2 unless otherwise mentioned.}) and total allocation time of $n$ probes~\cite{RS98}.  On the other hand, using $d$ probes per ball and placing the ball in the lightest bin, i.e. $\gd{d}$, first studied by Azar et. al \cite{ABKU99},  decreases the maximum load to $\frac{\ln \ln n}{\ln d} + O(1)$ with allocation time of $nd$. In other words, using $d \geq 2$ choices improves the maximum load exponentially, at a linear allocation cost.  

V\"ocking \cite{V03} introduced a slightly asymmetric algorithm, \lt{d}, which quite surprisingly guaranteed (w.h.p.) a maximum load of $\frac{\ln \ln n}{d \ln \phi_d} + O(1)$ (where $\phi_d$ is a constant between 1.61 and 2 when $d \geq 2$) using the same allocation time of $nd$ probes as \gd{d} when $m=n$. This analysis of maximum load for \gd{d} and \lt{d} was extended to the heavily-loaded case (when $m \gg n$) by Berenbrink et al. \cite{BCSV06}.  However, \lt{d} in \cite{V03} and \cite{BCSV06} utilizes additional processing.  Bins are initially sorted into groups and treated differently according to group membership.  Thus practical implementation, especially in distributed settings,  requires significant computational effort in addition to the probes themselves.  
\\
\\
\noindent {\bf Our Contribution.} 
We present a new algorithm, \fd{d}.  This algorithm requires no pre-sorting of bins; instead \fd{d} uses real-time feedback to adjust its number of probes for each ball\footnote{Thus we are concerned with the \emph{average} number of probes per ball throughout this paper.}.  

The natural comparison is with the classic $\gd{d}$ algorithm; $\fd{d}$ uses the same number of probes, on average, as $\gd{d}$ but produces a significantly smaller maximum load.  In fact, we show that the maximum load is as small as that of $\lt{d}$ when $m=n$.  Furthermore, it is comparable to $\lt{d}$ when heavily loaded.  For both the $m=n$ and heavily loaded cases, $\fd{d}$ has much lower computational overhead than $\lt{d}$.  

This simpler implementation makes \fd{d} especially suitable for practical applications; it is amenable to parallelization, for example, and requires no central control or underlying structure.  Some applications have a target maximum load and aim to minimize the necessary number of probes.  From this perspective, our algorithm again improves on $\gd{d}$: the maximum load of $\fd{\ln d}$ is comparable to that of $\gd{d}$, and uses exponentially fewer probes per ball.


\begin{duplicate}[Theorem~\ref{the:finite-light}]
  Use $\fd{d}$, where maximum number of probes allowed per ball is $2^{2d/3}$, to allocate $n$ balls into $n$ bins. The average number of probes required per ball is at most $d$ on expectation and w.h.p. Furthermore, the maximum load of any bin is at most $\frac{\log \log n}{0.66d} + O(1)$ with high probability when $d \geq 4$ and $n \geq \text{max}(2, n_0)$, where $n_0$ is the smallest value of $n$ such that for all $n > n_0$, $36 \log n \left(  \frac{72 e \log n}{5n} \right)^4 \leq \frac{1}{n^2}$.
\end{duplicate}

\begin{duplicate}[Theorem~\ref{the:finite-heavy}]
Use \fd{d}, where maximum number of probes allowed per ball is $2^{d/2.17}$, to allocate $m$ balls into $n$ bins. When $m \geq 72 (n \lambda \log n + n)$ where $\lambda$ is taken from Lemma~\ref{lem:initial-gap}, $n \geq n_0$ where $n_0$ is the smallest value of $n$ that satisfies $0.00332 n (\lambda \log n + 1) d/2^{d/2.17} \geq \log n$, and $d \geq 6$, it takes at most $d$ probes on average to place every ball on expectation and with high probability. Furthermore, for an absolute constant~$c$,
	\small
	\begin{center}
	$\prob\left(\text{Max. load of any bin} > \frac{m}{n} + \frac{\log \log n}{ 0.46d} + c \log \log \log n\right) \leq c (\log \log n)^{-4}$.
	\end{center}
	\normalsize
\end{duplicate}

Our technique for proving that the average number of probes is bounded is novel to the best of our knowledge.  As the number of probes required by each ball is dependent on the configuration of the balls-in-bins at the time the ball is placed, the naive approach to computing its expected value quickly becomes too conditional.  Instead, we show that this conditioning  can be eliminated by  carefully overcounting the number of probes required for each configuration, leading to a proof that is then quite simple.   The heavily-loaded case is significantly more complex than the $m=n$ case; however the basic ideas remain the same.

The upper bound on the maximum load is proved using the layered induction technique.  However, because \fd{d} is a dynamic algorithm, the standard recursion used in layered induction must be altered.  We use coupling and some more complex analysis to adjust the standard layered induction to this context.

We furthermore provide a tight lower bound on the maximum load for a broad class of algorithms which use variable probing.


\begin{duplicate}[Theorem~\ref{the:class1-lower-bound}]
	Let \alg{k} be any algorithm that places $m$ balls into $n$ bins, where $m \geq n$, sequentially one by one and satisfies the following conditions:
	\begin{enumerate}
		\item At most $k$ probes are used to place each ball.
		\item For each ball, each probe is made uniformly at random to one of the $n$ bins.
		\item For each ball, each probe is independent of every other probe.
	\end{enumerate}
	The maximum load of any bin after placing all $m$ balls using \alg{k} is at least $\frac{m}{n} + \frac{\ln \ln n}{\ln k} - \Theta(1)$ with high probability.
\end{duplicate}

We use the above theorem to provide a lower bound on the maximum load of \fd{d}, which is tight up to $O(\log \log \log n)$ terms.

\begin{duplicate}[Theorem~\ref{the:fd-lower-bound}]
The maximum load of any bin after placing $m$ balls into $n$ bins using \fd{d}, where maximum number of probes allowed per ball is $2^{\Theta(d)}$, is at least $\frac{m}{n} + \frac{\ln \ln n}{\Theta(d)} - \Theta(1)$ with high probability.
\end{duplicate}
~\\
\noindent {\bf Related Work.} 
Several other  algorithms in which the number of probes performed by each ball is adaptive in nature have emerged in the past, such as work done by Czumaj and Stemann~\cite{CS01} and by Berenbrink et al.~\cite{BKSS13}.

Czumaj and Stemann~\cite{CS01} present an interesting ``threshold" algorithm.  First they define a process \emph{Adaptive-Allocation-Process}, where each load value has an associated threshold and a ball is placed when the number of probes that were made to find a bin of a particular load exceeded the associated threshold. Then, by carefully selecting the thresholds for the load values, they develop \emph{M-Threshold}, where each ball probes bins until it finds one whose load is within some predetermined threshold.  The bounds on maximum load and on average allocation time are better than our algorithm's, but the trade-off is that computing the required threshold value often depends on the knowledge (typically unavailable in practical applications) of the total number of balls that will ever be placed. Furthermore their proofs are for $m=n$ and don't extend easily to when $m>n$. 

More recently, Berenbrink et al.~\cite{BKSS13} develop a new threshold algorithm, \emph{Adaptive}, which is similar to \emph{M-Threshold} but where the threshold value used for a given ball being placed depends on the number of balls placed thus far. They analyze this algorithm when $m \geq n$ and also extend the analysis of \emph{M-Threshold} from \cite{CS01} to the $m \geq n$ case. They show that both algorithms have good bounds on maximum load and average allocation time, but again this comes at the trade-off of requiring some sort of global knowledge when placing the balls. In the case of \emph{Adaptive}, each ball must know what its order in the global placement of balls is, and in the case of \emph{M-Threshold}, each ball must know the total number of balls that will ever be placed. Our algorithm is unique in that it  requires no such global knowledge at all; it is able to make decisions based on the probed bins' load values alone.
~\\

\noindent \textbf{Definitions.} 
In the course of this paper we will use several terms from probability theory, which we define below for convenience.

Consider two Markov chains $A_t$ and $B_t$ over time $t \geq 0$ with state spaces $S_1$ and $S_2$ respectively. A \textit{coupling} (cf. \cite{MU05}) of $A_t$ and $B_t$ is a Markov chain ($A_t$, $B_t$) over time $t \geq 0$ with state space $S_1 \times S_2$ such that $A_t$ and $B_t$ maintain their original transition probabilities.

Consider two vectors $u,v \in \mathbb{Z}^n$. Let $u'$ and $v'$ be permutations of $u$ and $v$ respectively such that $u'_i \geq u'_{i+1}$ and $v'_i \geq v'_{i+1}$ for all $1 \leq i \leq n-1$. We say $u$ \textit{majorizes} $v$ (or $v$ \textit{is majorized by} $u$) when
\begin{center}
	$\sum \limits_{j=1}^{i} u'_j \geq \sum \limits_{j=1}^{i} v'_j, \forall 1 \leq i \leq n$. 
\end{center}
For a given allocation algorithm $C$ which places balls into $n$ bins, we define the \textit{load vector} $u^t \in {(\mathbb{Z}^*)}^n$ of that process after $t$ balls have been placed as follows: the $i^{\text{th}}$ index of $u^t$ denotes the load of the $i^{\text{th}}$ bin (we can assume a total order on the bins according to their IDs). Note that $u^t$, $t \geq 0$, is a Markov chain.

Consider two allocation algorithms $C$ and $D$ that allocate $m$ balls. Let the load vectors for $C$ and $D$ after $t$ balls have been placed using the respective algorithms be $u^t$ and $v^t$ respectively. We say that $C$ \textit{majorizes} $D$ (or $D$ \textit{is majorized by} $C$) if there is a coupling between $C$ and $D$ such that $u^t$ majorizes $v^t$ for all $0 \leq t \leq m$.

Berenbrink et al. \cite{BCSV06} provide an illustration of the above ideas being applied in the load balancing context.
~\\

We also use Theorem~2.1 from Janson~\cite{J14} in order to achieve high probability concentration bounds on geometric random variables. We first set up the terms in the theorem and then restate it below. Let $X_1, \ldots, X_n$ be $n \geq 1$ geometric random variables with parameters $p_1, \ldots, p_n$ respectively. Define $p_* = \min_{i} p_i$, $X = \sum \limits_{i=1}^{n} X_i$, and $\mu = E[X] = \sum \limits_{i=1}^n \frac{1}{p_i}$. Now we have the following lemma.

\begin{lemma}[Theorem 2.1 in \cite{J14}] \label{lem:janson-the}
	For any $p_1, \ldots, p_n \in (0,1]$ and any $\Lambda \geq 1$, $Pr(X \geq \Lambda \mu) \leq e^{- p_*  \mu (\Lambda - 1 - \ln \Lambda)}$.
\end{lemma}

\noindent \textbf{Organization of Paper.} 
The structure of this paper is as follows. In Section~\ref{sec:model}, we define the model formally and present the $\fd{d}$ algorithm. We then analyze the algorithm when $m=n$ in Section~\ref{sec:finite} and give a proof that the total number of probes used by $\fd{d}$ to place $n$ balls is $nd$ with high probability, while the maximum bin load is still upper bounded by $\frac{\log \log n}{0.66d} + O(1)$ with high probability. We provide the analysis of the algorithm when $m>n$ in Section~\ref{sec:finite-heavy}, namely that the number of probes is on average $d$ per ball with high probability and the maximum bin load is upper bounded by $\frac{m}{n} + \frac{\log \log n}{0.46d} + O(\log \log \log n)$ with probability close to $1$. We provide a matching lower bound for maximum bin load tight up to the $O(\log \log \log n)$ term for algorithms with variable number of probes and \fd{d} in particular in Section~\ref{sec:lower-bound}. In Section~\ref{sec:experiments}, we give experimental evidence that our $\fd{d}$ algorithm indeed results in a maximum load that is comparable to $\lt{d}$ when $m=n$. Finally, we provide some concluding remarks and scope for future work in Section~\ref{sec:conc}.


\section{The \fd{d}  Algorithm}
\label{sec:model}
The idea behind this algorithm is to use probes more efficiently.  In the standard $d$-choice model, effort is wasted in some phases.  For example, early on in the distribution, most bins have size 0 and there is no need to search before placing a ball.   On the other hand, more effort in other phases would lead to significant improvement.  For example, if $.9n$ balls have been distributed, most bins already have size at least 1 and thus it is harder to avoid creating a bin of size 2.  \fd{d} takes this variation into account by probing until it finds a difference, then making its decision.

This algorithm uses probes more efficiently than other, fixed-choice algorithms, while still having a balanced outcome.  Each ball probes at most $2^{\Theta(d)}$ bins (where $d \geq 6$ and by extension $2^{\Theta(d)} > 2$) uniformly at random until it has found two bins with different loads (or a bin with zero load) and places the ball in the least loaded of the probed bins (or the zero loaded bin). If all $2^{\Theta(d)}$ probed bins are equally loaded, the ball is placed (without loss of generality) in the last probed bin.  The pseudocode for \fd{d} is below. Note that we use the $\Theta()$ to hide a constant value. The exact values are different for $m=n$ and $m \gg n$ and are $2/3$ and $1/2.17$ respectively.

\small
\alglanguage{pseudocode}
\begin{algorithm}
	\caption{$\fd{d}$  }
	\label{alg:firstdiff}
	(Assume $2^{\Theta(d)} > 2$. The following algorithm is executed for each ball.)
	\begin{algorithmic}[1]
		\BlockOn {Repeat $2^{\Theta(d)}$ times}
			\State Probe a new bin chosen uniformly at random
			\If{The probed bin has zero load}
				\State Place the ball in the probed bin and exit
			\EndIf
			\If{The probed bin has load that is different from those probed before}
				\State Place the ball in the least loaded bin (breaking ties arbitrarily) and exit
			\EndIf
		\BlockOff
		\State Place the ball in the last probed bin
	\end{algorithmic}
\end{algorithm}
\normalsize

As we can see, the manner in which a ball can be placed using \fd{d} can be classified as follows:

\begin{enumerate}
	\item The first probe was made to a bin with load zero.
	\item All probes were made to bins of the same load.
	\item One or more probes were made to bins of larger load followed by a probe to a bin of lesser load.
	\item One or more probes were made to bins of lesser load followed by a probe  to a bin of larger load.
\end{enumerate}


\section{Analysis of $\fd{d}$ when $m=n$}
\label{sec:finite}

\begin{theorem}\label{the:finite-light}
	Use $\fd{d}$, where maximum number of probes allowed per ball is $2^{2d/3}$, to allocate $n$ balls into $n$ bins. The average number of probes required per ball is at most $d$ on expectation and w.h.p. Furthermore, the maximum load of any bin is at most $\frac{\log \log n}{0.66d} + O(1)$ with high probability when $d \geq 4$ and $n \geq \text{max}(2, n_0)$, where $n_0$ is the smallest value of $n$ such that for all $n > n_0$, $36 \log n \left(  \frac{72 e \log n}{5n} \right)^4 \leq \frac{1}{n^2}$.
\end{theorem}

\begin{proof}
	First, we  show that an upper bound on the average number of probes per ball is $d$ on expectation and w.h.p. Subsequently, we  show that the maximum load at the end of placing all $n$ balls is as desired w.h.p.
\\\\	
\subsection{Proof of Number of Probes}
\begin{lemma}\label{lem:finite-light-num-probes}
		The number of probes required to place $m=n$ balls into $n$ bins using $\fd{d}$, where maximum number of probes allowed per ball is $2^{2d/3}$, is at most $nd$ on expectation and with high probability when $d \geq 4$.
\end{lemma}
\begin{proof}
		Let $k$ be the maximum number of probes allowed to be used by \fd{d} per ball, i.e. $k = 2^{2d/3}$. We show that the total number of probes required to place all balls does not exceed $1.5 n \log k$ w.h.p. and thus $nd$ probes are required to place all balls.

		Let the balls  be indexed from 1 to $n$ in the order in which they are placed.  Our analysis proceeds in two phases. For a value of $T$ that will be fixed subsequently, the first $T+1$ balls are analyzed in the first phase and remaining balls are analyzed in the second. Consider the ball indexed by $t$, $1 \le t \le n$. Let $X_t$ be the random variable denoting the number of probes it takes for \fd{d} to place ball $t$.

		\noindent \textbf{Phase One:} $t\leq T+1$. We couple \fd{d} with the related process that probes until it finds a difference in bin loads or runs out of probes, without treating empty bins as special; in other words, the \fd{d} algorithm without lines~3 and~4. One additional rule for the related process is that if an empty bin is probed first, then after the process finishes probing, the ball will be placed in that first probed bin, i.e. the empty bin. Note that this is a valid coupling; if an empty bin is probed then under both \fd{d} and this process the ball is placed in an empty bin, and if no empty bin is probed the two processes are exactly the same. Let $Y_t$ be the number of probes required by this related process to place ball $t$ in the configuration where there are $t-1$ bins of load 1 and $n-t+1$ bins of load 0. Notice that for any configuration of balls in bins, $X_t \leq Y_t$; furthermore, the configuration after placement under both \fd{d} and this new process is the same. You can see this by a simple sequence of couplings.
		
		First, choose some arbitrary configuration with $n\alpha_i$ bins of size $i$ for $i=0, 1, 2, \ldots$.  That configuration will be probed until bins of two different sizes are discovered, i.e. until the set probed intersects two distinct $\alpha_i$ and $\alpha_j$. Couple this with the configuration that has $\sum_{i=1}^{n} n\alpha_i$ bins of size 1 and the rest of size 0.  This configuration requires more probes than the original configuration; it continues until the set probed intersects $\alpha_0$ and $\alpha_{\neq 0}$.  Second, note that the configuration with $t-1$ bins of size 1 and $n-t+1$ bins of size 0 requires even more probes than this one. This is because restricting the bins to size 1 can only decrease the number of empty bins.  Finally, note that the ball's placement under either \fd{d} or this new process leads to the same configuration at time $t+1$ (up to isomorphism); if \fd{d} places a ball in an empty bin, so does this process.
		
		We first derive the expected value of $Y_t$. The expected number of probes used by \fd{d} is upper bounded by the expected number of probes until a size-0 bin appears, i.e. the expected number of probes used by the \fd{d} algorithm without line 1. This is of course $n/(n-t+1)$.  The overall expected number of probes for the first $T+1$ steps is
		\begin{align*}
			E[Y] &= E[\sum_{t=1}^{T+1} Y_t] \\
			&= \sum_{t=1}^{T+1} E[Y_t] \\
			&\leq \sum_{t=1}^{T+1}\frac{n}{n-t+1} \\
			&=n\sum_{i=n-T}^{n}\frac{1}{i}\\
			&~\sim n(\log n-\log (n-T)) \\
			&=n\log \left(\frac{n}{n-T}\right).
		\end{align*}
		
		Now we will find $T$ such that the expected number of probes in phase one, $E[Y]$, is $n \log k$, i.e. 
		$$n\log \left(\frac{n}{n-T}\right) = n \log k.$$  Solving, we get 
		$T=n(1-1/k).$  
		Now, recall that we want a high probability bound on the number of probes required to place each ball in Phase One when running \fd{d}, i.e. $\sum_{t=0}^{n(1-1/k)} X_t$. Recall that $X_t \leq Y_t$, and as such a high probability bound on $\sum_{t=0}^{n(1-1/k)} Y_t$ suffices. We can now use Lemma~\ref{lem:janson-the} with
		$\Lambda = 1.01$, $\mu = n\log k$, and $p_* = \frac{1}{k}$.
		\begin{align*}
			Pr\left(\sum_{t=0}^{n(1-1/k)} X_t \geq 1.01 n \log k\right) &\leq Pr\left(\sum_{t=0}^{n(1-1/k)} Y_t \geq 1.01 n \log k\right)\\
			&\leq e^{-\frac{1}{k} \cdot (n \log k) \cdot (1.01 - 1 - \ln 1.01)}\\
			&\leq O\left(\frac{1}{n} \right) \text{ since } k << n 
		\end{align*}
		
		\noindent \textbf{Phase Two:} $t>T+1$.  Rather than analyzing in detail, we use the fact that the number of probes for each ball is bounded by $k$, i.e. $X_t \leq k$, $\forall t > T+1$. So the number of probes overall in this phase is at most \[k(n-T-1)=k(n-n(1-1/k)-1)=n-k.\]
		
		So the total number of probes w.h.p. is $1.01 n \log k + n - k \leq 1.5 n \log k$ (when $k = 2^{2d/3}$ and $d \geq 4$). When $k = 2^{2d/3}$, an upper bound on the number of probes to place all $n$ balls is $nd$ probes on expectation and w.h.p., as desired. 
	\end{proof}
	
	\subsection{Proof of Maximum Load} \label{maxload}
	\begin{lemma}\label{lem:finite-light-max-load}
		The maximum load in any bin after using $\fd{d}$, where maximum number of probes allowed per ball is $2^{2d/3}$, to allocate $n$ balls into $n$ bins is at most $\frac{\log \log n}{0.66d} + O(1)$ with high probability when $d \geq 4$ and  $n \geq \text{max}(2, n_0)$, where $n_0$ is the smallest value of $n$ such that for all $n > n_0$, $36 \log n \left(  \frac{72 e \log n}{5n} \right)^4 \leq \frac{1}{n^2}$.
	\end{lemma}

	\begin{proof}
		While the proof follows along the lines of the standard layered induction argument~\cite{MU05,TW14}, we have to make a few non-trivial adaptations to fit our context where the number of probes is not fixed.
		
		Let $k$ be the maximum number of probes allowed to be used by \fd{d} per ball, i.e. $k = 2^{2d/3}$. Define $v_i$ as the fraction of bins of load at least $i$ after $n$ balls are placed. Define $u_i$ as the number of balls of height at least $i$ after $n$ balls are placed. It is clear that $v_i * n \leq u_i$.
		
		We wish to show that the $\prob(\text{Max. load} \geq \frac{\log \log n}{\log k} + \gamma) \leq \frac{1}{n^c}$ for some constants $\gamma \geq 1$ and $c \geq 1$. Set $i^* = \frac{\log \log n}{\log k} + 11$ and $\gamma = 15$. Equivalently, we wish to show that $\prob(v_{i^*+4} > 0) \leq \frac{1}{n^c}$ for some constant $c \geq 1$.
		
		In order to aid us in this proof, let us define a non-increasing series of numbers $\beta_{11}, \beta_{12}, \ldots, \beta_{i^*}$ as upper bounds on $v_{11}, v_{12}, \ldots v_{i^*}$. Let us set $\beta_{11} = \frac{1}{11}$.
		
		\noindent Now,
		\begin{align}
			\prob(v_{i^*+4} > 0)
			&= \prob(v_{i^*+4} > 0 | v_{i^*} \leq \beta_{i^*}) \prob(v_{i^*} \leq \beta_{i^*}) \nonumber\\ &\hspace{2em}+ \prob(v_{i^*+4} > 0 | v_{i^*} > \beta_{i^*}) \prob(v_{i^*} > \beta_{i^*}) \nonumber \\
			&\leq \prob(v_{i^*+4} > 0 | v_{i^*} \leq \beta_{i^*}) + \prob(v_{i^*} > \beta_{i^*}) \nonumber \\
			&= \prob(v_{i^*+4} > 0 | v_{i^*} \leq \beta_{i^*}) \nonumber\\ &\hspace{2em} + \prob(v_{i^*} > \beta_{i^*} | v_{i^*-1} \leq \beta_{i^* - 1}) \prob(v_{i^* - 1} \leq \beta_{i^* - 1}) \nonumber\\& \hspace{2em}+ \prob(v_{i^*} > \beta_{i^*} | v_{i^*-1} > \beta_{i^* - 1}) \prob(v_{i^* - 1} > \beta_{i^* - 1}) \nonumber \\
			&\leq \prob(v_{i^*+4} > 0 | v_{i^*} \leq \beta_{i^*}) \nonumber \\ &\hspace{2em}+ \sum\limits_{i=12}^{i^*} \prob(v_{i} > \beta_{i} | v_{i-1} \leq \beta_{i-1}) + \prob(v_{11} > \beta_{11}) \label{eq:main-eq}
		\end{align}
		
		Here, $\prob(v_{11} > \beta_{11}) = 0$. It remains to find upper bounds for the remaining two terms in the above equation.
		
		We now derive a recursive relationship between the $\beta_i$'s for $i \geq 11$. $\beta_{i+1}$ acts as an upper bound for the fraction of bins of height at least $i+1$ after $n$ balls are placed. In order for a ball placed to land up at height at least $i+1$, one of 3 conditions must occur:
		\begin{itemize}
			\item All $k$ probes are made to bins of height at least $i$.
			\item Several probes are made to bins of height at least $i$ and one is made to a bin of height at least $i+1$.
			\item One probe is made to a bin of height at least $i$ and several probes are made to bins of height at least $i+1$.
		\end{itemize}
		
		Thus the probability that a ball is placed at height at least $i+1$, conditioning on $v_j\leq \beta_j$ for $j\leq i+1$ at that time, is
		\begin{align*}
			&\leq \beta_i^k + \beta_i\beta_{i+1}\left(1 + \beta_i + \beta_i^2 + \ldots + \beta_i^{k-2}\right) \\ &\hspace{2em}+ \beta_i \beta_{i+1} \left(1 + \beta_{i+1} + \beta_{i+1}^2 + \ldots + \beta_{i+1}^{k-2} \right) \\
			&\leq \beta_i^k + \beta_i \beta_{i+1} \left( \frac{1 - \beta_i^{k-1}}{1 - \beta_i} + \frac{1 - \beta_{i+1}^{k-1}}{1 - \beta_{i+1}} \right) \\
			&\leq \beta_i^k + \beta_{11} \beta_{i+1} \left( 2 * \frac{1}{1 - \beta_{11}}\right) \\
			&\leq \beta_i^k + \frac{2 \beta_{i+1}}{10}
		\end{align*}
		
		Let $v_{i+1}(t)$ be the fraction of bins with load at least $i+1$ after the $1\leq t\leq n$ ball is placed in a bin.
		
		Let $t^*=\min[\arg\min_t v_{i+1}(t)>\beta_{i+1}, n]$, i.e. $t^*$ is the first $t$ such that $v_{i+1}(t)>\beta_{i+1}$ or $n$ if there is no such $t$.  The probability that $t^*<n$ is bounded by the probability that a binomial random variable $B(n, \beta_i^k + \frac{2 \beta_{i+1}}{10})$ is greater than $\beta_{i+1}n$.
		
		Fix $\beta_{i+1} = \frac{10}{3} \beta_i^k \geq \frac{2 n (\beta_i^k + \frac{2 \beta_{i+1}}{10})}{n}$. Then 
		using a Chernoff bound, we can say that with high probability, $t^*=n$ or $v_{i+1} \leq \beta_{i+1}$, so long as $e^{- \frac{\left(\beta_i^k + \frac{2 \beta_{i+1}}{10}\right)}{3}} = O(\frac{1}{n^c})$ for some constant $c \geq 1$. 
		
		Now, so long as $\beta_{i+1} \geq \frac{18 \log n}{n}$, $e^{- \frac{\left(\beta_i^k + \frac{2 \beta_{i+1}}{10}\right)}{3}} = O(\frac{1}{n^c})$.
		Notice that at $i = i^*$, the value of $\beta_i$ dips below  $\frac{18 \log n}{n}$. This can be seen by solving the recurrence with $\log \beta_{11} = - \log 11$ and $\log \beta_{i+1} = \log (\frac{10}{3})  + k \log \beta_i$.
		
		\begin{align*}
			\log \beta_{i^*} &= \log \left(\frac{10}{3}\right) (1 + k + k^2 + \ldots + k^{\log_k \log n - 1}) \\&\hspace{2em} + k^{\log_k log n} (- \log 11) \\
			&= \log \left(\frac{10}{3}\right) \left(\frac{k^{\log_k \log n} - 1}{k-1}\right) - (\log n) (\log 11) \\
			&\leq (\log n) \left(\log \left(\frac{10}{3}\right) - \log 11\right) \\
			&\leq -1.7 \log n
		\end{align*}
		
		Therefore, as it is, $\beta_{i^*} \leq \frac{1}{n^{1.7}} \leq \frac{18 \log n}{n}$ when $n \geq 2$. In order to keep the value of $\beta_i$ at least at $\frac{18 \log n}{n}$, we set
		
		\begin{align}
			\beta_{i+1} = \text{max}\left(\frac{10}{3} \beta_i^k, \frac{18 \log n}{n}\right)
		\end{align}
		
		With the values of $\beta_i$ defined, we proceed to bound $\prob(v_i > \beta_i | v_{i-1} \leq \beta_{i-1}), \forall 12 \leq i \leq i^*$.
		\noindent For a given $i$,
		
		\begin{align*}
			\prob(v_i > \beta_i | v_{i-1} \leq \beta_{i-1}) & = \prob(n v_i > n \beta_i | v_{i-1} \leq \beta_{i-1}) \\
			& \leq \prob(u_i > n \beta_i | v_{i-1} \leq \beta_{i-1}) \\
		\end{align*}
		
		We upper bound the above inequality using the following idea. Let $Y_r$ be an indicator variable set to 1 when the following 2 conditions are met: (i) the $r^{\mbox{th}}$ ball placed is of height at least $i$ and (ii) $v_{i-1} \leq \beta_{i-1}$. $Y_r$ is set to 0 otherwise. Now for all $1 \leq r \leq n$, the probability that $Y_r = 1$ is upper bounded by $\beta_{i-1}^k + \frac{2 }{10} \beta_i \leq \frac{3}{10} \beta_i + \frac{2 }{10}\beta_i \leq \frac{\beta_i}{2}$. Therefore, the probability that the number of balls of height at least $i$ exceeds $\beta_i$ is upper bounded by $\prob(B(n, \frac{\beta_i}{2}) > n \beta_i)$, where $B(\cdot,\cdot)$ is a binomial random variable with given parameters.
		
		Recall the Chernoff bound, for $0 < \delta \leq 1, \prob( X \geq (1 + \delta) \mu) \leq e^{- \frac{\mu \delta^2}{3}}$, where $X$ is the sum of independent Poisson trials and $\mu$ is the expectation of $X$. If we set $\delta = 1$, then we have
		
		\begin{align*}
			\prob(v_i > \beta_i | v_{i-1} \leq \beta_{i-1}) & \leq \prob( B(n, \frac{\beta_i}{2}) > n \beta_i) \\
			& \leq e^{- \dfrac{n \cdot (\frac{\beta_i}{2})}{3}} \\
			& \leq e^{ - \dfrac{n \cdot (\frac{18 \log n}{n})}{6}}  \mbox{ (since } \beta_i \geq \frac{18 \log n}{n}, \forall i \leq i^*) \\
			& \leq \frac{1}{n^3}
		\end{align*}
		
		\noindent Thus we have
		
		\begin{align}
			& \sum\limits_{j=12}^{i^*} \prob(v_j > \beta_j | v_{j-1} \leq \beta_{j-1}) \leq \frac{\log \log n}{n^3} \nonumber \\
			\implies & \sum\limits_{j=l+1}^{i^*} \prob(v_j > \beta_j | v_{j-1} \leq \beta_{j-1}) \leq \frac{1}{2n^2}  \mbox{ (since } n \geq 2) \label{eq:middle}
		\end{align}

		\noindent Finally, we need to upper bound $\prob(v_{i^*+4} > 0 | v_{i^*} \leq \beta_{i^*})$. Consider a particular bin of load at least $i^*$. Now the probability that a ball will fall into that bin is 
		\begin{align*}
			&\leq \frac{1}{n} \cdot \left( \beta_{i^*}^{k-1} + 2 \beta_{i^*+1} \left( \frac{1}{1 - \beta_{i^*+1}^k} \right) \right) \\
			&\leq \frac{1}{n} \cdot \left( \beta_{i^*}^{k-1} + 2 \beta_{i^*+1} \left(\frac{1}{1 - \beta_{11}} \right)\right) \\
			&\leq \frac{1}{n} \cdot \left( \beta_{i^*}^{k-1} + \frac{22}{10} \beta_{i^*}\right) (\text{since } \beta_i \text{ is a non-increasing function}) \\
			&\leq \frac{1}{n} \cdot \frac{32}{10} \cdot \beta_{i^*} \mbox{ (since } k \geq 2 \mbox{ and } \beta_{i^*} \leq 1 )
		\end{align*}
		
		Now, we upper bound the probability that 4 balls fall into a given bin of load at least $i^*$ and then use a union bound over all the bins of load at least $i^*$ to show that the probability that the fraction of bins of load at least $\beta_{i^* + 4}$ exceeds 0 is negligible.
		
		First, the probability that 4 balls fall into a given bin of load at least $\beta_{i^*}$ is
		\begin{align*}
			&\leq \prob(B(n, (\frac{1}{n} \cdot \frac{32}{10} * \beta_{i^*})) \geq 4 )\\
			&\leq \binom {n}4 \left(\frac{1}{n} \cdot \frac{32}{10} \cdot \beta_{i^*} \right)^4 \\
			&\leq \left( e \cdot n \cdot \left(\frac{1}{n} \cdot \frac{32}{10} \cdot \beta_{i^*} \right)\cdot \frac{1}{4} \right)^4 \\
			&\leq \left( \frac{32}{10} \cdot  \frac{e  \beta_{i^*}}{4} \right)^4
		\end{align*}
		
		Taking the union bound across all possible $\beta_{i^*} n$ bins, we have the following inequality
		
		\begin{align}
			&\prob(v_{i^*+4} > 0 | v_{i^*} \leq \beta_{i^*}) \leq (\beta_{i^*} n) \cdot \left(  \frac{32}{10} \cdot \frac{e \beta_{i^*}}{4} \right)^4 \nonumber \\
			\implies &\prob(v_{i^*+4} > 0 | v_{i^*} \leq \beta_{i^*})  \leq (18 \log n) \cdot \left( \frac{32}{10} \cdot  \frac{18 e \log n}{4n} \right)^4 \nonumber \\
			\implies &\prob(v_{i^*+4} > 0 | v_{i^*} \leq \beta_{i^*}) \leq \frac{1}{2n^2} \mbox{ (since } n \geq n_0) \label{eq:end}
		\end{align}
		
		Putting together equations~\ref{eq:main-eq}, \ref{eq:middle}, and \ref{eq:end}, we get
		
		\begin{align*}
			\prob(v_{i^*+4} > 0 ) &\leq \frac{1}{2n^2} + \frac{1}{2n^2} \\
			&\leq \frac{1}{n^2}
		\end{align*}
		
		\noindent Thus 
		
		\begin{align*}
			\prob\left(\text{Max. Load} \geq \frac{\log \log n}{\log k} + 15\right) & = \prob(v_{i^*+4} > 0) \\
			&\leq \frac{1}{n^2}
		\end{align*} 
	\end{proof}
	
	\noindent From Lemma~\ref{lem:finite-light-num-probes} and Lemma~\ref{lem:finite-light-max-load}, we immediately arrive at Theorem~\ref{the:finite-light}.
\end{proof}


\section{Analysis of $\fd{d}$ when $m \gg n$}
\label{sec:finite-heavy}

\begin{theorem} \label{the:finite-heavy}
Use \fd{d}, where maximum number of probes allowed per ball is $2^{d/2.17}$, to allocate $m$ balls into $n$ bins. When $m \geq 72 (n \lambda \log n + n)$ where $\lambda$ is taken from Lemma~\ref{lem:initial-gap}, $n \geq n_0$ where $n_0$ is the smallest value of $n$ that satisfies $0.00332 n (\lambda \log n + 1) d/2^{d/2.17} \geq \log n$, and $d \geq 6$, it takes at most $d$ probes on average to place every ball on expectation and with high probability. Furthermore, for an absolute constant~$c$,
	\small
	\begin{center}
	$\prob\left(\text{Max. load of any bin} > \frac{m}{n} + \frac{\log \log n}{ 0.46d} + c \log \log \log n\right) \leq c (\log \log n)^{-4}$.
	\end{center}
	\normalsize
\end{theorem}

\begin{proof}

First we show that the average number of probes per ball is at most $d$ on expectation and w.h.p. We then show the maximum load bound holds w.h.p. 

\subsection{Proof of Number of Probes}
~\\\textbf{Remark:} The earlier version of this paper \cite{AMRU16} had a different proof in this subsection. The overall idea of overcounting the number of probes remains the same but the specific argument of how to justify and go about such an overcounting is changed and cleaner now. More specifically, we have replaced Lemmas~4.2,~4.3, and~4.4 with the argument that follows the header ``Overcounting method" and concludes at the header ``Expectation bound". Also note that Lemma~4.6 from our earlier version is no longer required due to the way we've constructed our argument.

The main difficulty with analyzing the number of probes comes from the fact that the number of probes needed for each ball depends on where each of the previous balls were placed. Intuitively, if all the previous balls were placed such that each bin has  the same number of balls, the number of probes will be $2^{d/2.17}$. On the other hand, if  a significant number of bins are at different load levels, then, the ball will be placed with very few probes.  One might hope to prove that the system always displays a variety of loads, but unfortunately, the system (as we verified experimentally) oscillates between being very evenly loaded and otherwise. Therefore, we have to take a slightly more nuanced approach that takes into account that the number of probes cycles between high (i.e. as high as $2^{d/2.17}$) when the loads are even and as low as 2 when there is more variety in the load.

\begin{lemma}\label{lem:finite-heavy-num-probes}
When $m \geq 72 (n \lambda \log n + n)$ where $\lambda$ is taken from Lemma~\ref{lem:initial-gap}, $n \geq n_0$ where $n_0$ is the smallest value of $n$ that satisfies $0.00332 n (\lambda \log n + 1) d/2^{d/2.17} \geq \log n$, and $d \geq 6$, using \fd{d}, where maximum number of probes allowed per ball is $2^{d/2.17}$, takes at most $md$ probes on expectation and with high probability to place the $m$ balls in $n$ bins.
\end{lemma}

\begin{proof}

Let the maximum number of probes allowed per ball using \fd{d} be $k$, i.e. $k = 2^{d/2.17}$. Throughout this proof, we will assume that the maximum load of any bin is at most $m/n + \lambda \log n$, which holds with high probability owing to Lemma~\ref{lem:initial-gap}. The low probability event that the maximum load  exceeds $m/n + \lambda \log n$ will contribute very little to the overall number of probes because the probability that any ball exceeds a height of $m/n + \lambda \log n$ will be an arbitrarily small inverse polynomial in $n$. Therefore, such a ball will contribute $o(1)$ probes to the overall number of probes even when we liberally account $k$ probes for each such ball (as long as $k \ll n$). Let $m$ balls be placed into $n$ bins; we assume $m \geq 72 (n \lambda \log n + n)$.

In order to prove the lemma, we proceed in three stages. In the first stage, we consider an arbitrary sequence of placing $m$ balls into the $n$ bins. We develop a method that allows us to overcount the number of probes required to place the ball at each step of this placement. In the second stage, we proceed to calculate the expected number of probes required to place the balls. Finally in the third stage, we show how to get a high probability bound on the number of probes required to place each ball.

~\\
\noindent \textbf{Overcounting method}\\
First, we couple the process of \fd{d} with a similar process where the zero bin condition to place balls is not used. In this similar process, if a zero bin is probed first, more probes are made until either a bin with a different load is probed or until all $k$ probes are made. Then, the ball is placed in the first bin probed (the zero bin). In case the first bin probed is not a zero bin, then the process acts exactly as \fd{d}. It is clear that this process will take more probes than \fd{d}  while still making the same placements as \fd{d}. Thus, any upper bounds on number of probes obtained for this process apply to \fd{d}. For the remainder of this proof, we analyze this process.

Now we describe the method we use to overcount the number of probes. Consider an arbitrary sequence of placing $m$ balls into $n$ bins. We will describe a method to associate each configuration that arises from such a placement with a ``canonical" configuration, which we define later in this proof. Each such canonical configuration requires more probes to place the ball than the actual configuration. We ensure that the mapping of actual configurations to canonical configurations is a one-to-one mapping. Thus, by counting the number of probes required to place a ball in every possible canonical configuration, we overcount the number of probes required to place every ball. 

Imagine coupling according to probe sequences.  That is, consider all $[n]^{k}\times [n]^k \times \ldots \times [n]^k=[n]^{km}$ possible sequences of probes.  Each timestep corresponds to a particular length-$k$ sequence of bin labels $[n]$, which direct the bins to be probed.  Note that the probe sequences are each equally likely, and that they fully determine the placement.

For each probe sequence $S$, let $X^S$ be the sequence of configurations generated by $S$.  Let $x^S$ be the sequence of numbers-of-probes-used at each timestep of $X^S$.  (Note that $S$ gives $k$ potential probes at each timestep, but the entries of $x^S$ are often less than $k$, as often the ball is placed without using the maximum number of probes.) For convenience, we drop the superscript $S$ subsequently and use $X$ and $x$ to refer to the vectors.

Now we give a few definitions. Consider a particular bin with balls placed in it, one on top of another. If exactly $\ell -1 $ balls are placed below a given ball, we say that ball is at height $\ell$ or level $\ell$. Balls of the same height are said to be of the same level. We say a level contains $b$ balls if there are $b$ balls at that level. A level containing $n$ balls is said to be complete. A level containing at least one ball but less than $n$ balls is said to be incomplete. For a given configuration, consider the highest complete level $\ell$. For the given configuration, we define a plateau as any level $\geq \ell$ with at least one ball at that level. Intuitively, the plateaus for the configuration are the highest complete level and any higher incomplete levels. Notice that it is possible that a given configuration may only have one plateau when there are no incomplete levels. Further, notice that for a given configuration if two plateaus exist at levels $\ell$ and $\ell+2$ with number of balls $b_1$ and $b_2$, it implies that there exists a plateau at level $\ell+1$ with number of balls in $[b_2,b_1]$.

Consider a particular configuration $X_i$ in the sequence $X$. Call its number of plateaus~$p$. Consider the plateaus in increasing order of their levels and call them $\ell_1, \ell_2, \ldots \ell_p$.  Call the number of balls at each level $\ell_i$, $b_i$.

We now define the canonical configuration $C_{\ell, b}$.  This is the configuration with no balls at level greater than $\ell$, $b$ balls at level $\ell$, and $n$ balls at every level less than $\ell$.  

We will associate each configuration $X_i$ with a set of canonical configurations $\mathcal{C}_i$. For each plateau of $X_i$ at level $\ell_j$, include the canonical configuration $C_{\ell_j, b_j}$ in the set.

Note that for any probe sequence (not just the specific sequence $S$), the number of probes utilized by $X_i$ (e.g. $x_i$ for the sequence $S$) is less than or equal to the number of probes utilized by any of the configurations in $\mathcal{C}_i$.  Therefore the expected number of probes used to place a ball in configuration $X_i$ is less than or equal to the expected number of probes used to place a ball in any of the configurations in $\mathcal{C}_i$.  We now describe a way to select one particular configuration out of $\mathcal{C}_i$ to associate with each configuration $X_i$. We choose this configuration such that the mapping between all configurations $X_i$, $1 \leq i \leq m$, and the selected canonical configurations will be a one-to-one mapping. Furthermore, we show that every selected canonical configurations will be unique from the others and thus the set of all selected canonical configurations will not be a multiset. Thus, by counting the number of probes required to place balls in every possible canonical configuration, we overcount the number of probes required to place balls in any sequence $S$.

Look at the canonical configurations associated with configurations $X_i$, $1 \leq i \leq m$ over some entire sequence.  Let $\mathcal{C}_{0}$ refer to the set of canonical configurations before any ball is placed and let it be the empty set. The set of canonical configurations $\mathcal{C}_{i-1}$ differs from $\mathcal{C}_{i}$ in at most three configurations.  Consider that $i-1$ balls have been placed and there are now $p$ plateaus with levels $\ell_i$, $1 \leq i \leq p$ and corresponding $b_i$, $1 \leq i \leq p$ values. 

\begin{itemize}
	\item If the $i^{th}$ ball is placed at level $\ell_j$, level $\ell_{j+1}$ exists ($\geq 1$ balls are present at level $\ell_{j+1}$), and $b_{j+1}\neq n-1$, then $\mathcal{C}_{i}=(\mathcal{C}_{i-1} \backslash \{C_{\ell_{j+1}, b_{j+1}}\} ) \cup \{C_{\ell_{j+1}, b_{j+1}+1}\}$. 
	\item If level $\ell_{j+1}$ exists and $b_{j+1}=n-1$, then $\mathcal{C}_{i}=(\mathcal{C}_{i-1} \backslash \{C_{\ell_j, b_j}, C_{\ell_{j+1}, b_{j+1}}\}) \cup \{C_{\ell_{j+1}, n}\}$. 
	\item If level $\ell_{j+1}$ does not exist, then $\mathcal{C}_{i} = \mathcal{C}_{i-1} \cup\{C_{\ell_{j}+1, 1}\}$.
\end{itemize}    

Notice that in every scenario, there is exactly one configuration that is added to $\mathcal{C}_{i-1}$ to get $\mathcal{C}_{i}$. We denote the newly-added configuration as the \emph{selected} canonical configuration of $\mathcal{C}_i$. 

Now, for a given sequence of configurations $X_i$, $1 \leq i \leq m$ it is clear that each selected canonical configuration is uniquely chosen. We show in the following lemma, Lemma~\ref{lem:not-multiset}, that every selected canonical configuration in the sequence is unique and different from the other selected canonical configurations. In other words, the set of selected canonical configurations will not be a multiset. Furthermore, it has earlier been established that for any configuration and one of its canonical configurations, it takes at least as many probes to place a ball in the latter as it does in the former. Thus, by calculating the number of probes it would take to place balls in all possible canonical configurations $\bigcup_{i=1}^{m} \mathcal{C}_i$, we can overcount the number of probes required to place all balls into bins.

\begin{lemma}\label{lem:not-multiset}
The set of selected canonical configurations for any sequence of configurations $X_i$, $1 \leq i \leq m$ will not be a multiset.
\end{lemma}

\begin{proof}
Consider an arbitrary sequence and within it an arbitrary configuration $X_i$ for some $1 \leq i \leq m$. To reach this configuration, a ball was placed previously in some level $\ell - 1$ and extended level $\ell$ from $b$ balls to $b+1$ balls, $0 \leq b \leq n-1$. The selected canonical configuration for $X_i$ will be $C_{\ell, b+1}$. Since balls can only be added and never deleted, once a level is extended to some $b+1$ number of balls, placing another ball can never extend that same level to $b+1$ balls. That level can only be henceforth extended to a larger number of balls up to $n$ balls. Thus a given configuration $C_{\ell,b+1}$ will never appear twice in the set of selected canonical configurations.
\end{proof}

~\\
\noindent \textbf{Expectation bound}\\
As mentioned earlier, we assume the maximum load of any bin is $m/n + \lambda \log n$. Thus, for any given sequence of placements, the final configuration never has any balls at level $m/n + \lambda \log n +1$. Thus, when calculating the number of probes taken to place all balls, we need only consider the number of probes required to place a ball in every canonical configuration $C_{\ell, b}$, $0 \leq \ell \leq m/n + \lambda \log n$, $0 \leq b \leq n-1$.

For a given canonical configuration $C_{\ell, b}$ let $Y_{\ell,b}$ be a random variable denoting the number of probes required to place the ball using \fd{d} without the zero bin condition. Let $Y$ be a random variable denoting the total number of probes required to place a ball in each of the possible canonical configurations. Thus $Y = \sum_{\ell=1}^{m/n + \lambda \log n} \sum_{b=0}^{n-1} Y_{\ell,b}$.

For a given configuration $C_{\ell,b}$, a ball is placed when either it first hits bins of level $\ell$ several times and then a bin of level $\ell-1$, it hits bins of level $\ell-1$ several times and then a bin of level $\ell$, or it makes $k$ probes. Thus, using geometric random variables, we see that $E[Y_{\ell,b}] = \min (b/(n-b) + (n-b)/b, k)$. For the first $\frac{n}{k}$ and last $\frac{n}{k} - 1$ canonical configurations for a given level, let us give away the maximum number of probes, i.e. $Y_{\ell,b} = k$ for any $\ell$ and for $0 \leq b \leq \frac{n}{k} - 1$ and $n - \frac{n}{k} +1 \leq b \leq n-1$. We now want to calculate the expected number of probes for the middle canonical configurations.


\noindent  Therefore

\begin{align*}
\sum \limits_{\ell = 0}^{\frac{m}{n} + \lambda \log n} \sum\limits_{b=\frac{n}{k}}^{n - \frac{n}{k}} E[Y_{\ell,b}] &=\left(\frac{m}{n} + \lambda \log n + 1 \right)\left( \sum\limits_{b=\frac{n}{k}}^{n - \frac{n}{k}} \frac{n-b}{b} + \sum\limits_{b=\frac{n}{k}}^{n - \frac{n}{k}} \frac{b}{n-b} \right)\\
&= \left(\frac{m}{n} + \lambda \log n + 1 \right) \left(n \sum\limits_{b=\frac{n}{k}}^{n - \frac{n}{k}} \frac{1}{b} - \sum\limits_{b=\frac{n}{k}}^{n - \frac{n}{k}} \frac{b}{b} + n \sum\limits_{b=\frac{n}{k}}^{n - \frac{n}{k}} \frac{1}{n-b} - \sum\limits_{b=\frac{n}{k}}^{n - \frac{n}{k}} \frac{n-b}{n-b} \right)\\
&= \left(\frac{m}{n} + \lambda \log n + 1 \right) \left(n \sum\limits_{b=\frac{n}{k}}^{n - \frac{n}{k}} \frac{1}{b} + n \sum\limits_{y = n - \frac{n}{k}}^{\frac{n}{k}} \frac{1}{y}  - 2n + \frac{4n}{k} \right)\\
&= \left(\frac{m}{n} + \lambda \log n + 1 \right) \left( 2n \left(\sum\limits_{b=\frac{n}{k}}^{n - \frac{n}{k}} \frac{1}{b} -1 + \frac{2}{k}\right) \right)\\
&\approx  2 \left(m + n\lambda \log n + n \right)\left(\log \left(n - \frac{n}{k}\right) - \log \left(\frac{n}{k}\right) -1 + \frac{2}{k} \right)\\
&=  2 \left(m + n \lambda \log n + n \right) \left( \log (k-1)-1 + \frac{2}{k} \right)
\end{align*}

~\\
\noindent \textbf{High probability bound}\\
Now, we may apply Lemma~\ref{lem:janson-the} with $\Lambda = 1.01$, $\mu$ taken from above, and $p_* = \frac{1}{k}$.

\begin{align*}
&Pr\left(\sum \limits_{\ell = 0}^{\frac{m}{n} + \lambda \log n} \sum \limits_{b=\frac{n}{k}}^{n-\frac{n}{k}} Y_{\ell,b} > 2.02 (m + n \lambda \log n + n)(\log (k-1)-1 + \frac{2}{k})\right) \\
&\leq e^{-\frac{1}{k} \cdot 2 (m + n \lambda \log n + n)(\log (k-1)-1 + \frac{2}{k}) \cdot (1.01 - 1 - \ln 1.01)} \\
&\leq O\left(\frac{1}{n}\right) \text{ (since } k > 2 \text{ and } n \geq n_0)  
\end{align*}

 Therefore, with high probability, the total number of probes 

\begin{align*}
Y &= \sum \limits_{\ell=0}^{\frac{m}{n} + \lambda \log n} \sum \limits_{b=0}^{n-1} Y_{\ell,b} \\
&\leq \sum \limits_{\ell=0}^{\frac{m}{n} + \lambda \log n}\sum \limits_{b=0}^{\frac{n}{k}-1} Y_{\ell,b} + \sum \limits_{\ell=0}^{\frac{m}{n} + \lambda \log n}\sum \limits_{b=\frac{n}{k}}^{n - \frac{n}{k}} (Y_{\ell,b})  + \sum \limits_{\ell=0}^{\frac{m}{n} + \lambda \log n}\sum \limits_{b=n - \frac{n}{k} + 1}^{n-1} Y_{\ell,b} \\
&\leq \left( \frac{m}{n} + \lambda \log n + 1 \right) k\frac{n}{k} + 2.02 \left(m + n \lambda \log n + n \right) \left( \log (k-1)-1 + \frac{2}{k} \right) + \left( \frac{m}{n} + \lambda \log n + 1 \right)k \left(\frac{n}{k} - 1\right)\\
&\leq 2.14 (m + n \lambda \log n + n)\log k \\
&\leq 2.17 m \log k \text{ (since } m \geq 72 (n \lambda \log n + n) \text{)}
\end{align*}


\noindent Thus when $m$ balls are placed into $n$ bins, an upper bound on both the expected total probes and the total probes with high probability is $2.17 m \log k$. Therefore on expectation and with high probability, the number of probes per ball is at most $d$ since $k = 2^{d/2.17}$.
\end{proof}

\subsection{Proof of Maximum Load}
\begin{lemma}\label{lem:finite-heavy-max-load}
Use \fd{d}, where maximum number of probes allowed per ball is $2^{d/2.17}$, to allocate $m$ balls into $n$ bins. For any $m$, for an absolute constant $c$,
\small
\begin{center}
$\prob\left(\text{Max. load of any bin} > \frac{m}{n} + \frac{\log \log n}{ 0.46d} + c \log \log \log n\right) \leq c (\log \log n)^{-4}$.
\end{center}
\normalsize
\end{lemma}

\begin{proof}
This proof follows along the lines of that of Theorem 2 from \cite{TW14}.  In order to prove Lemma~\ref{lem:finite-heavy-max-load}, we make use of a theorem from \cite{PTW10} which gives us an initial, loose, bound on the gap $G^t$ between the maximum load and average load for an arbitrary $m$. We then use a lemma to tighten this gap. We use one final lemma to show that if this bound on the gap holds after all $m$ balls are placed, then it will hold at any time prior to that.

First, we establish some notation. Let $k$ be the maximum number of probes permitted to be made per ball by \fd{d}, i.e. $k = 2^{d/2.17}$. After placing $nt$ balls, let us define the \emph{load vector} $X^t$ as representing the difference between the load of each bin and the average load (as in \cite{BCSV06,PTW10}). Without loss of generality we order the individual values of the vector in non-increasing order of load difference, i.e. $X_1^t \geq X_2^t \geq \ldots \geq X_n^t$. So $X_i^t$ is the load in the $i^{\mbox{th}}$ most loaded bin minus $t$. For convenience, denote $X_1^t$ (i.e. the gap between the heaviest load and the average) as $G^t$.\\

\noindent \textbf{Initial bound on gap}\\
We now give an upper bound for the gap between the maximum loaded bin and the average load after placing some arbitrary number of balls $nt$. In other words, we show $\prob(G^t \geq x)$ is negligible for some $x$.  This $x$ will be our initial bound on the gap $G^t$.

\begin{lemma}\label{lem:initial-gap}
For arbitrary constant $c$, after placing an arbitrary $nt$ balls into bins under \fd{d}, there exist some constants $a$ and $b$ such that $\prob(G^t \geq \frac{c \log n}{a}) \leq \frac{bn}{n^c}$. Thus there exists a constant $\lambda$ that gives $\prob(G^t \geq \lambda \log n) \leq \frac{1}{n^c}$ for a desired $c$ value. 
\end{lemma}
In order to prove Lemma \ref{lem:initial-gap}, we need two additional facts.  The first is the following basic observation:  
\begin{lemma}\label{lem:fd-maj}
\fd{d} is majorized by \gd{2} when $d \geq 2$.
\end{lemma}

\begin{proof}
Let the load vectors for \fd{d} and \gd{2} after $t$ balls have been placed using the respective algorithms be $u^t$ and $v^t$ respectively. Now we follow the standard coupling argument (refer to Section 5 in \cite{BCSV06} for an example). Couple \fd{d} with \gd{2} by letting the bins probed by \gd{2} be the first $2$ bins probed by \fd{d}. We know that \fd{d} makes at least 2 probes when $d \geq 2$. It is clear that \fd{d} will always place a ball in a bin with load less than or equal to that of the bin chosen by \gd{2}. This ensures that if majorization was preserved prior to the placement of the ball, then the new load vectors will continue to preserve majorization; again, see \cite{BCSV06} for a detailed example. Initially, $u^0$ is majorized by $v^0$ since both vectors are the same. Using induction, it can be seen that if $u^t$ is majorized by $v^t$ at the time the $t^{\text{th}}$ ball was placed, it would continue to be majorized at time $t+1$, $0 \leq t \leq m-1$. 
Therefore, \fd{d} is majorized by \gd{2} when $d \geq 2$.
\end{proof}

The other fact is the following theorem about \gd{d} taken from \cite{PTW10} (used similarly in \cite{TW14} as Theorem~3).

\begin{theorem}\cite{PTW10}\label{the:base-case} 
Let $Y^t$ be the load vector generated by \gd{d}.  Then for every $d>1$ there exist positive constants $a$ and $b$ such that for all $n$ and all $t$, 
\begin{center}
$E\left( \sum\limits_i e^{a |Y_i^t|} \right) \leq bn$.
\end{center}
\end{theorem}

\noindent We are now ready to prove Lemma \ref{lem:initial-gap}.
\begin{proof}[Proof of Lemma~\ref{lem:initial-gap}.]
 Combining Lemma \ref{lem:fd-maj} with Theorem \ref{the:base-case} tells us that, if $X^t$ is the load vector generated by \fd{d}, $$E\left( \sum\limits_i e^{a |X_i^t|} \right) \leq bn.$$
Clearly, $\prob(G^t \geq \frac{c \log n}{a}) = \prob(e^{a G^t} \geq n^c)$. Observe that $\sum\limits_i e^{a |X_i^t|}\geq e^{aG^t}$.  Then
\begin{align*}
\prob(G^t \geq \frac{c \log n}{a}) &= \prob(e^{a G^t} \geq n^c) \\
  &\leq \frac{E[e^{a G^t}]}{n^c} \mbox{ (by Markov's inequality)}\\
  &\leq \frac{bn}{n^c} \mbox{ (by Theorem~\ref{the:base-case} and Lemma \ref{lem:fd-maj})}
\end{align*}
 and the theorem is proved.
\end{proof}

\noindent \textbf{Reducing the gap}\\
Lemma \ref{lem:initial-gap} gives an initial bound on $G^t$ of order $\log n$.  The next step is to reduce it to our desired gap value. For this reduction, we use a modified version of Lemma 2 from \cite{TW14}, with a similar but more involved proof. We now give the modified lemma and prove it. 

\begin{lemma}\label{lem:reduce-gap}
For every $k$, there exists a universal constant $\gamma$ such that the following holds: for any $t,\ell, L$ such that $1 \leq \ell \leq L \leq n^{\frac{1}{4}}$, $L = \Omega(\log \log n)$ and $\prob(G^t \geq L) \leq \frac{1}{2}$,
\begin{center}
$\prob(G^{t+L} \geq \frac{\log \log n}{\log k} + \ell + \gamma) \leq \prob(G^t \geq L) + \frac{16 b L^3}{e^{a\ell}} + \frac{1}{n^2}$,
\end{center}
where $a$ and $b$ are the constants from Theorem~\ref{the:base-case}.
\end{lemma}

\begin{proof}
This proof consists of many steps.  We first observe that Lemma~\ref{lem:reduce-gap} follows directly from  Lemma~\ref{lem:initial-gap} for sufficiently small $n$.  We then 
use layered induction to bound the proportion of bins of each size for larger $n$.  This in turn allows us to compute our desired bound on the probability of a large gap occurring.\\

\noindent \textbf{Proof of Lemma~\ref{lem:reduce-gap} for smaller values of $n$}\\
Define $n_1$ to be the minimum value of $n$ such that $L \geq 2$ (recall $L = \Omega(\log \log n)$). Define $n_2$ to be the minimum value of $n$ such that\\ $ (18 \log n) * \left( \frac{18 e n^{\frac{1}{4}} \log n}{n} \right)^4 \leq \frac{1}{2n^2}$. Define $n_3$ to be the minimum value of $n$ such that $n \geq 54 \log n$. Define absolute constant $n_0 = max(n_1,n_2, n_3)$. \\
Notice that, when $n \leq n_0$, Lemma~\ref{lem:initial-gap} implies that Lemma~\ref{lem:reduce-gap} holds with $\gamma = O(\log n_0)$.
If $n\leq n_0$, then $$\prob(G^{t+\ell}\geq \log \log n +\ell+\gamma) \leq \prob(G^{t+L} \geq \gamma).$$  Consider the right hand side of Lemma~\ref{lem:reduce-gap}.  $$\prob(G^t\geq L)+\frac{16bL^3}{\exp(a\ell)}+\frac{1}{n^2}\geq n^{-2},$$ so it will be sufficient to prove the inequality $$\prob(G^{t+L} \geq \gamma) \leq n^{-2}.$$
Since there are no conditions on $t$ in Lemma~\ref{lem:initial-gap}, we may rewrite it as  $$\prob(G^{t+L}\geq \lambda \log n) \leq n^{-c}.$$ Let $c=2$ and compute the constant $\lambda$ accordingly.  Set $\gamma = \lambda \log n_0\geq \lambda \log n$.  Then $$n^{-2}\geq \prob(G^{t+L} \geq \lambda \log n)\geq \prob(G^{t+L} \geq \gamma),$$ and we are done.\\

\noindent \textbf{Rewriting initial probability inequality}\\
We now prove Lemma~\ref{lem:reduce-gap} assuming $n > n_0$. Start by rewriting the probability in terms of $\prob(G^t\geq L)$.

\begin{align*}
 \prob(G^{t+L}\geq \frac{\log \log n}{\log k} + \ell + \gamma) &= \prob(G^{t+L} \geq \frac{\log \log n}{\log k} + \ell + \gamma | G^t \geq L) \prob(G^t \geq L)\\&\hspace{2em} + \prob(G^{t+L} \geq \frac{\log \log n}{\log k} + \ell + \gamma | G^t < L) \prob(G^t < L)\\
  &\leq \prob(G^t \geq L) + \prob(G^{t+L} \geq \frac{\log \log n}{\log k} + \ell + \gamma | G^t < L)
\end{align*}

To prove the theorem, then, it is enough to to show that \\$\prob(G^{t+L} \geq \frac{\log \log n}{\log k} + \ell + \gamma | G^t < L) \leq \frac{16 b L^3}{e^{a\ell}} + \frac{1}{n^2}$.\\\\\\
\noindent \textbf{Bins' loads}\\
Define $v_i$ to be the fraction of bins of load at least $t+L+i$ after $(t+L)n$ balls are placed. Let us set $i^* = \frac{\log \log n}{\log k} + \ell$ and set $\gamma = 4$. Using this new notation, we want to show that $\prob(G^{t+L} \geq i^* + 4| G^t<L)$ is negligible. This can be thought of as showing that the probability of the fraction of bins of load at least $t+ L + i^* + 4$ exceeding $0$ after $(t+L)n$ balls are placed, conditioned on the event that $G^t<L$, is negligible. 

Suppose we have a non-increasing series of numbers $\beta_0, \beta_1, \ldots, \beta_i, \ldots$ that are upper bounds for $v_0, v_1, \ldots, v_i, \ldots$.  Then we know that
\begin{align*}
\prob(v_{i^*+4} > 0) &= \prob(v_{i^*+4} > 0 | v_{i^*} \leq \beta_{i^*}) \prob(v_{i^*} \leq \beta_{i^*})  + \prob(v_{i^*+4} > 0 | v_{i^*} > \beta_{i^*}) \prob(v_{i^*} > \beta_{i^*}) \\
  &\leq \prob(v_{i^*+4} > 0 | v_{i^*} \leq \beta_{i^*}) + \prob(v_{i^*} > \beta_{i^*}) \\
  &\leq \prob(v_{i^*+4} > 0 | v_{i^*} \leq \beta_{i^*}) + \sum\limits_{j=\ell+1}^{i^*} \prob(v_j > \beta_j | v_{j-1} \leq \beta_{j-1}) + \prob(v_{\ell} > \beta_{\ell}) \text{ (successively expanding} \\ &\hspace{1em} \text{ and bounding the } \prob(v_{i^ *} > \beta_{i^*}) \text{ term and its derivatives)} 
\end{align*}

\noindent Conditioning both sides on $G^t<L$, we have
\begin{align}\label{eq:heavy-main-eq}
\prob(v_{i^*+4} > 0 | G^t<L) &  \leq \prob(v_{i^*+4} > 0 | v_{i^*} \leq \beta_{i^*}, G^t<L)  \nonumber \\
 & \hspace{2em}+ \sum\limits_{i=\ell+1}^{i^*} \prob(v_i > \beta_i | v_{i-1} \leq \beta_{i-1}, G^t<L) \nonumber\\
 & \hspace{2em}+ \prob(v_{\ell} > \beta_{\ell} | G^t<L)
\end{align}
It remains to find appropriate $\beta_i$ values.  We use a layered induction approach to show that $v_i$'s don't exceed the corresponding $\beta_i$'s with high probability.  This then allows us to upper bound each of the 3 components of equation~\ref{eq:heavy-main-eq}. \\

\noindent \textbf{Base case of layered induction}\\
In order to use layered induction, we need a base case. Let us set $\beta_{\ell} = \frac{1}{8L^3}$, for the $\ell$ in the statement of the theorem. Now,
\begin{align*}
\prob(v_{\ell} > \beta_{\ell} | G^t<L) & = \frac{\prob((v_{\ell} > \frac{1}{8 L^3}) \bigcap (G^t<L))}{\prob(G^t<L)} \\
  &\leq 2 \cdot \prob(v_{\ell} > \frac{1}{8 L^3}) \mbox{ (since, by the statement of the theorem } \prob(G^t<L) \geq \frac{1}{2}) \\
  &\leq 2 \cdot \frac{8bL^3}{e^{a\ell}} \mbox{ (applying Markov's inequality and using Theorem~\ref{the:base-case}) } \\
  &\leq \frac{16 b L^3}{e^{a\ell}}
\end{align*}

\noindent Therefore we have the third term of Equation~\ref{eq:heavy-main-eq} bounded: 
\begin{equation}\label{eq:beginning}
\prob(v_{\ell} > \beta_{\ell} | G^t<L) \leq \frac{16 b L^3}{e^{a\ell}}
\end{equation}

\noindent \textbf{Recurrence relation for layered induction}\\
We now define the remaining $\beta_i$ values recursively. Note that for all $i \geq \ell$, $\beta_i \leq \beta_{\ell}$. Let $u_i$ be defined as the number of balls of height at least $t+L+i$ after $(L+t)n$ balls are placed. 

Initially there were $nt$ balls in the system. Then we threw another $nL$ balls into the system. Remember that $t+L$ is the average load of a bin after $nL$ balls are further placed. Because we condition on $G^t<L$, we have it that any ball of height $i$, $i \geq 1$, must have been one of the $nL$ balls placed. 

Therefore the number of bins of load $t+L+i + 1$ after $(t+L)n$ balls are placed is upper bounded by the number of balls of height at least $t+L+i+1$. So $v_{i+1} n \leq u_{i+1}$. In order to upper bound $v_{i+1}$, we can upper bound $u_{i+1}$. 

Recall the algorithm places a ball in a bin of load $t+L+i$ if it probes $k$ times and sees a bin of load $t+L+i$ each time; or if it probes $j<k$ times and sees a bin of load $t+L+i$ each time, then probes a bin of load $\geq t+L+i+1$; or if it probes $j<k$ times and sees a bin of load at least $t+L+i+1$ each time (where the load of the bin probed each time is the same), then probes a bin of load $t+L+i$.  Thus the probability that a ball will end up at height at least $t+L+i+1$ is

\begin{align*}
&\leq \beta_i^k + \beta_i\beta_{i+1}\left(1 + \beta_i + \beta_i^2 + \ldots + \beta_i^{k-2}\right) + \beta_i \beta_{i+1} \left(1 + \beta_{i+1} + \beta_{i+1}^2 + \ldots + \beta_{i+1}^{k-2} \right) \\
&\leq \beta_i^k + \beta_i \beta_{i+1} \left( \frac{1 - \beta_i^{k-1}}{1 - \beta_i} + \frac{1 - \beta_{i+1}^{k-1}}{1 - \beta_{i+1}} \right) \\
&\leq \beta_i^k + \beta_l \beta_{i+1} \left( 2 * \frac{1}{1 - \beta_l}\right) \\
&\leq \beta_i^k + \frac{2 \beta_{i+1}}{8 L^3 - 1}
\end{align*}

Let $v_{i+1}(f)$ be the fraction of bins with load at least $t+i+1$ after the $tn + f^{\text{th}}$
, $1\leq f\leq nL$, ball is placed in a bin.

Let $f^*=\min[\arg\min_f v_{i+1}(f)>\beta_{i+1}, nL]$, i.e. $f^*$ is the first $f$ such that $v_{i+1}(f)>\beta_{i+1}$ or $nL$ if there is no such $f$. By our preceding argument, the probability that $f^*<nL$ is bounded by the probability that a binomial
random variable $B(nL, \leq \beta_i^k + \frac{2 \beta_{i+1}}{8 L^3 - 1})$ is greater than $\beta_{i+1}nL$.

Fix $$\beta_{i+1} = 2L \frac{8L^3 - 1}{8L^3 - 4L -1} \beta_i^k \geq \frac{2 nL (\beta_i^k + \frac{2 \beta_{i+1}}{8 L^3 - 1})}{n}.$$ 
 Then using a Chernoff bound, we can say that with high probability, $f^*=nL$ or $v_{i+1} \leq \beta_{i+1}$, so long as $e^{- \frac{\left(\beta_i^k + \frac{2 \beta_{i+1}}{8 L^3 - 1}\right)}{3}} = O(\frac{1}{n^c})$ for some constant $c \geq 1$. 

Now, so long as $\beta_{i+1} \geq \frac{18 \log n}{n}$, $e^{- \frac{\left(\beta_i^k + \frac{2 \beta_{i+1}}{8 L^3 - 1}\right)}{3}} = O(\frac{1}{n^c})$.   In other words, this upper bound holds for the placement of all $nt+nL$ balls.\\ 

We now show that according to the previous recurrence relation, $\beta_{i^*}$ dips below $\frac{18 \log n}{n}$. We later propose a modified recurrence relation which sets the value of $\beta_i$ to the maximum of the value of obtained from the recurrence and $\frac{18 \log n}{n}$. This ensures that $\beta_{i^*} = \frac{18 \log n}{n}$. This upper bound will be used later in the argument. We have, from the value of $\beta_{\ell}$ and the above discussion, 
\begin{center}
$\log \beta_{\ell} = - 3 \log (2L)$ and\\
$\log \beta_{i+1} = k \log \beta_i + \log(2L) + \log (\frac{8L^3 - 1}{8L^3 - 4L -1})$
\end{center}

\noindent Solving the recursion for $\log \beta_{\ell + \log \log n}$, we get 
\begin{align*}
\log \beta_{\ell + \log \log n} &= \frac{k^{\log \log n} - 1}{k-1} \log \left(\frac{2L (8L^3 - 1)}{8L^3 - 4L -1}\right)  - 3 k^{\log\log n} \log (2L) \\
 &\leq k^{\log \log n}\left( (-3k + 4) \log(2L) + \log \left( \frac{8L^3 - 1}{8L^3 - 4L - 1} \right) \right)\\
 &\leq k^{\log \log n}\left( (-6 + 4) \log(2L) + \log \left( \frac{8L^3 - 1}{8L^3 - 4L - 1} \right) \right) \\
 &\leq k^{\log \log n}\left( (-1.5) \log(2L) \right) \mbox{ (when } L \geq 2) \\
 &\leq 2^{\log \log n}\left( (-1.5) \log(2L) \right) \\
 &\leq (-1.5) (\log n) 
\end{align*}

\noindent Therefore, $\beta_{i^*} \leq n^{-1.5}$, when $L \geq 2$. Since $n \geq n_1$, we have $L \geq 2$.  Thus $\beta_{i^*} <\frac{18 \log n}{n}$, as desired. 

 Now, we need to bound $\prob(v_i > \beta_i | v_{i-1} \leq \beta_{i-1}, G^t<L)$ for all $i$'s from $\ell+1$ to $i^*$. Let us set $\beta_{i+1} = max(2L \frac{8L^3 - 1}{8L^3 - 4L -1} \beta_i^k, \frac{18 \log n}{n})$.

Using the values of $\beta_i$ generated above, we prove that for all $i$ such that $\ell+1 \leq i \leq i^*$, $\prob(v_i > \beta_i | v_{i-1} \leq \beta_{i-1}, G^t<L) \leq \frac{1}{n^3}$.
\\\\
\noindent For a given $i$,
\begin{align*}
\prob(v_i > \beta_i | v_{i-1} \leq \beta_{i-1}, G^t < L) & = \prob(n v_i > n \beta_i | v_{i-1} \leq \beta_{i-1}, G^t < L) \\
& \leq \prob(u_i > n \beta_i | v_{i-1} \leq \beta_{i-1}, G^t < L) \\
\end{align*}

We now upper bound the above inequality using the following idea. Let $Y_r$ be an indicator variable set to 1 when all three of the following conditions are met: (i) the $nt + r^{\mbox{th}}$ ball placed is of height at least $t + L + i$, (ii) $v_{i-1} \leq \beta_{i-1}$ and (iii) $G^t<L$. $Y_r$ is set to 0 otherwise. Now for all $1 \leq r \leq nL$, the probability that $Y_r = 1$ is upper bounded by $\beta_i^k + \frac{2 \beta_{i+1}}{8 L^3 - 1} \leq \frac{8L^3 - 1}{8L^3 - 4L -1} \beta_{i-1}^k \leq \frac{\beta_i}{2L}$. Since we condition on $G^t < L$, the number of balls of height at least $t + L$ or more come only from the $nL$ balls placed. Therefore, the probability that the number of balls of height at least $n + L + i$ exceeds $\beta_i$ is upper bounded by $\prob(B(nL, \frac{\beta_i}{2L}) > \beta_i)$, where $B(.,.)$ is a binomial
random variable with given parameters. 

According to Chernoff's bound, for $0 < \delta \leq 1, \prob( X \geq (1 + \delta) \mu) \leq e^{- \frac{\mu \delta^2}{3}}$, where $X$ is the sum of independent Poisson trials and $\mu$ is the expectation of $X$. If we set $\delta = 1$, then we have

\begin{align*}
\prob(v_i > \beta_i | v_{i-1} & \leq \beta_{i-1}, G^t < L)\\
 & \leq \prob( B(nL, \frac{\beta_i}{2L}) > \beta_i ) \\
& \leq e^{- \dfrac{n * (\frac{\beta_i}{2})}{3}} \\
& \leq e^{ - \dfrac{n * (\frac{18 \log n}{n})}{6}} \mbox{ (since } \beta_i \geq \frac{18 \log n}{n}, \forall i \leq i^*) \\
& \leq \frac{1}{n^3}
\end{align*}

\noindent Thus we bound the middle term in Equation~\ref{eq:heavy-main-eq}
\small
\begin{align}
& \sum\limits_{j=\ell+1}^{i^*} \prob(v_j > \beta_j | v_{j-1} \leq \beta_{j-1}, G^t < L) \leq \frac{\log \log n}{n^3} \nonumber\\
\implies & \sum\limits_{j=\ell+1}^{i^*} \prob(v_j > \beta_j | v_{j-1} \leq \beta_{j-1}, G^t < L) \leq \frac{1}{2n^2} \mbox{ (since } n \geq n_1) \label{eq:heavy-middle}
\end{align}
\normalsize
\noindent \textbf{Top layers of layered induction}\\
Finally, we need to upper bound the first term in Equation~\ref{eq:heavy-main-eq},\\ $\prob(v_{i^*+4} > 0 | v_{i^*} \leq \beta_{i^*}, G^t<L)$. Consider a bin of load at least $i^*$. We will upper bound the probability that a ball falls into this specific bin. Regardless of how the probes are made for that ball, one of them must be made to that specific bin. Thus we have a formula similar to our original recursion, but with a factor of $1/n$. 

Therefore the probability that a ball will fall into that bin is 
\begin{align*}
&\leq \frac{1}{n}\beta_i^{k-1} + \frac{1}{n}\beta_{i^*+1}\left(1 + \beta_{i^*} + \beta_{i^*} + \ldots + \beta_{i^*}^{k-2}\right) + \frac{1}{n} \beta_{i^*+1} \left(1 + \beta_{i^*+1} + \beta_{i^*+1}^2 + \ldots + \beta_{i^*+1}^{k-2} \right) \\
&\leq \frac{1}{n} \cdot \left( \beta_{i^*}^{k-1} + 2 \beta_{i^*+1} \left( \frac{1}{1 - \beta_{i^*+1}} \right) \right) \\
&\leq \frac{1}{n} \cdot \left( \beta_{i^*}^{k-1} + 2 \beta_{i^*} \left(\frac{1}{1 - \beta_{i^*}} \right)\right) \\
&\leq \frac{1}{n} \cdot \left( \beta_{i^*}^{k-1} + \frac{2n}{n-18 \log n} \beta_{i^*}\right)   \\
&\leq \frac{1}{n} \cdot \frac{3n - 18 \log n}{n - 18 \log n} \cdot \beta_{i^*} \mbox{ (since } k \geq 2 \mbox{ and } \beta_{i^*} \leq 1 ) \\
&\leq \frac{4}{n} \cdot \beta_{i^*} \text{ (since } n > n_3) \\
\end{align*}

Now, we upper bound the probability that 4 balls fall into a given bin of load at least $\beta_{i^*}$ and then use a union bound over all the bins of height at least $\beta_{i^*}$ to show that the probability that the fraction of bins of load at least $\beta_{i^* + 4}$ exceeds 0 is negligible.

First, the probability that 4 balls fall into a given bin of load at least $\beta_{i^*}$ is
\begin{align*}
&\leq \prob(B(nL, (\frac{4}{n} \cdot \beta_{i^*})) \geq 4 )\\
&\leq \binom {nL}4 \left(\frac{4}{n} \cdot \beta_{i^*} \right)^4 \\
&\leq \left( e \cdot nL \cdot (\frac{4}{n} \cdot \beta_{i^*} )\cdot \frac{1}{4} \right)^4 \\
&\leq \left( e L \beta_{i^*} \right)^4
\end{align*}

Taking the union bound across all possible $\beta_{i^*} n$ bins, we have the following inequality
\begin{align}
&\prob(v_{i^*+4} > 0 | v_{i^*} \leq \beta_{i^*}, G^t<L) \leq (\beta_{i^*} n) \cdot \left( e L \beta_{i^*}  \right)^4 \nonumber \\
\implies& \prob(v_{i^*+4} > 0 | v_{i^*} \leq \beta_{i^*}, G^t<L) \leq (18 \log n) \cdot \left( \frac{18 e L \log n}{n} \right)^4 \nonumber \\
\implies &\prob(v_{i^*+4} > 0 | v_{i^*} \leq \beta_{i^*}, G^t<L) \leq \frac{1}{2n^2} \mbox{ (since } n \geq n_2) \label{eq:heavy-end}
\end{align}

Putting together equations~\ref{eq:heavy-main-eq}, \ref{eq:beginning}, \ref{eq:heavy-middle}, and \ref{eq:heavy-end}, we get

\begin{align*}
\prob(v_{i^*+4} > 0 | G^t<L) &\leq \frac{16 b L^3}{e^{a\ell}} + \frac{1}{2n^2} + \frac{1}{2n^2} \\
&\leq \frac{16 b L^3}{e^{a\ell}} + \frac{1}{n^2}
\end{align*}

\noindent Thus 

\begin{align*}
\prob(G^{t+L} \geq \frac{\log \log n}{\log k} + \ell + 4 | G^t < L) & = \prob(v_{i^*+4} > 0 | G^t<L) \\
 &\leq \frac{16 b L^3}{e^{a\ell}} + \frac{1}{n^2}
\end{align*}

\noindent Finally
\begin{align*}
\prob(G^{t+L}\geq \frac{\log \log n}{\log k} + \ell + 4) &\leq \prob(G^t \geq L) + \prob(G^{t+L} \geq \frac{\log \log n}{\log k} + \ell + 4 | G^t < L) \\
&\leq \prob(G^t \geq L) + \frac{16 b L^3}{e^{a\ell}} + \frac{1}{n^2}
\end{align*}

\noindent Hence Lemma~\ref{lem:reduce-gap} is proved.
\end{proof}

By Lemma~\ref{lem:initial-gap}, we know that at some arbitrary time $t$, the gap will be $O(\log n)$ with high probability. Now applying Lemma~\ref{lem:reduce-gap} once with $L = O(\log n)$ and $\ell = O(\log \log n)$ with appropriately chosen constants, we get $\prob(G^{t+L}\geq \frac{\log \log n}{\log k} + O(\log \log n) + \gamma) \leq O((\log \log n)^{-4})$. Applying the lemma again with $L = O(\log \log n)$ and $\ell= O(\log \log \log n)$ with appropriately chosen constants, we get\\ $\prob(G^t > \frac{\log \log n}{\log k} + c \log \log \log n) \leq \frac{c}{(\log \log n)^4}$ when time $t = \omega(\log n)$.

We now show that as more balls are placed, the probability that the gap exceeds a particular value increases. This is true by Lemma 4 from \cite{TW14}:
\begin{lemma}\label{lem:time-reduction}\cite{TW14} 
For $t \geq t'$, $G^{t'}$ is stochastically dominated by $G^t$. Thus $E[G^{t'}] \leq E[G^t]$ and for every $z$, $\prob(G^{t'} \geq z) \leq \prob(G^t \geq z)$.
\end{lemma}

Although the setting is different in \cite{TW14}, their proof of Lemma \ref{lem:time-reduction} applies here as well.  Thus knowing the gap is large when time $t = \omega(\log n)$ with probability $O((\log \log n)^{-4})$, implies that for all values of $t' < t$, the gap exceeds the desired value with at most the same probability. Substituting $k = 2^{d/2.17}$ in $\prob(G^t > \frac{\log \log n}{\log k} + c \log \log \log n) \leq \frac{c}{(\log \log n)^4}$ and modifying the inequality to talk about max. load after $m$ balls have been thrown results in the lemma statement.

\noindent Thus concludes the proof of Lemma~\ref{lem:finite-heavy-max-load}.
\end{proof}

\noindent Putting together Lemma~\ref{lem:finite-heavy-num-probes} and Lemma~\ref{lem:finite-heavy-max-load}, we get Theorem~\ref{the:finite-heavy}. 
\end{proof}


\section{Lower Bound on Maximum Bin Load}
\label{sec:lower-bound}
 
 We now provide a lower bound to the maximum load of any bin after using \fd{d} as well as other types of algorithms which use a variable number of probes for Class 1 type algorithms as defined by V\"ocking \cite{V03}. Class 1 algorithms are those where for each ball, the locations are chosen uniformly and independently at random from the bins available. We first give a general theorem for this type of algorithm and then apply it to \fd{d}.\\
 
\begin{theorem} \label{the:class1-lower-bound}
Let \alg{k} be any algorithm that places $m$ balls into $n$ bins, where $m \geq n$, sequentially one by one and satisfies the following conditions:
\begin{enumerate}
\item At most $k$ probes are used to place each ball.
\item For each ball, each probe is made uniformly at random to one of the $n$ bins.
\item For each ball, each probe is independent of every other probe.
\end{enumerate}
The maximum load of any bin after placing all $m$ balls using \alg{k} is at least $\frac{m}{n} + \frac{\ln \ln n}{\ln k} - \Theta(1)$ with high probability.
\end{theorem}

\begin{proof}
We show that \gd{k} is majorized by \alg{k}, i.e. \gd{k} always performs better than \alg{k} in terms of load balancing. Thus any lower bound that applies to the max. load of any bin after using \gd{k} must also apply to \alg{k}.

Let the load vectors for \gd{k} and \alg{k} after $t$ balls have been placed using the respective algorithms be $u^t$ and $v^t$ respectively. We use induction on the number of balls placed to prove our claim of majorization. Initially, no ball is placed and by default $u^0$ is majorized by $v^0$. Assume that $u^{t-1}$ is majorized by $v^{t-1}$. We now use the standard coupling argument to prove the induction hypothesis. For the placement of the $t^{\text{th}}$ ball, let \alg{k} use $w_t$ probes. Couple \gd{k} with \alg{k} by letting the first $w_t$ bins probed by \gd{k} be the same bins probed by \alg{k}. \gd{k} will always make at least $w_t$ probes and thus possibly makes probes to lesser loaded bins than those probed by \alg{k}. Since \gd{k} places a ball into the least loaded bin it finds, it will place a ball into a bin with load at most the same as the one chosen by \alg{k}. Therefore $u^t$ is majorized by $v^t$. Thus by induction, we see that $u^t$ is majorized by $v^t$ for all $0 \leq t \leq m$. Therefore \gd{k} is majorized by \alg{k}.

It is known that the max. load of any bin after the placement of $m$ balls into $n$ bins ($m \geq n$) using \gd{k} is at least $\frac{m}{n} + \frac{\ln \ln n}{\ln k} - \Theta(1)$ with high probability \cite{BCSV06}. Therefore, the same lower bound also applies to \alg{k}.
\end{proof}

Now we are ready to prove our lower bound on the max. load of any bin after using \fd{d}.\\

\begin{theorem} \label{the:fd-lower-bound}
The maximum load of any bin after placing $m$ balls into $n$ bins using \fd{d}, where maximum number of probes allowed per ball is $2^{\Theta(d)}$, is at least $\frac{m}{n} + \frac{\ln \ln n}{\Theta(d)} - \Theta(1)$ with high probability.
\end{theorem}

\begin{proof}
We see that \fd{d} uses at most $2^{\Theta(d)}$ probes and satisfies the requirements of Theorem~\ref{the:class1-lower-bound}. Thus by substituting $k = 2^{\Theta(d)}$, we get the desired bound.
\end{proof}


\section{Experimental Results}
\label{sec:experiments}

\begin{table}[ht]
\caption{Experimental results for the maximum load for $n$ balls and $n$ bins based on 100 experiments for each configuration. Note that the maximum number of probes per ball in \fd{d}, denoted as $k$, is chosen such that the average number of probes per ball is fewer than $d$.}
\centering \vspace{1em}
\resizebox{1.0\columnwidth}{!}{%
\begin{tabular}{|c||c|c|c|c|c|c|c|c|c|}
  \hline
  & \multicolumn{3}{|c|}{$d$ = 2, $k$ = 3} & \multicolumn{3}{|c|}{$d$ = 3, $k$ = 10} & \multicolumn{3}{|c|}{$d$ = 4, $k$ = 30} \\
  \cline{2-10}
  n & $\gd{d}$ & $\lt{d}$ & $\fd{d}$ & $\gd{d}$ & $\lt{d}$ & $\fd{d}$ & $\gd{d}$ & $\lt{d}$ & $\fd{d}$ \\
  \hline \hline
  \multirow{3}{*}{$2^8$} & \textbf{2}...11\% & \textbf{2}...43\% & \textbf{2}...81\% & \textbf{2}...88\% & \textbf{2}...100\% & \textbf{2}...100\% & \textbf{2}...100\% & \textbf{2}...100\% & \textbf{2}...100\% \\
  & \textbf{3}...87\% & \textbf{3}...57\% & \textbf{3}...19\% & \textbf{3}...12\% &  &  &  &  &  \\
  & \textbf{4}... 2\% &  &  &  &  &  &  &  &  \\
  \hline\multirow{3}{*}{$2^{12}$} &  &  & \textbf{2}...10\% & \textbf{2}...12\% & \textbf{2}...96\% & \textbf{2}...100\% & \textbf{2}...93\% & \textbf{2}...100\% & \textbf{2}...100\% \\
  & \textbf{3}...99\% & \textbf{3}...100\% & \textbf{3}...90\% & \textbf{3}...88\% & \textbf{3}... 4\% &  & \textbf{3}... 7\% &  &  \\
  & \textbf{4}... 1\% &  &  &  &  &  &  &  &  \\
  \hline\multirow{3}{*}{$2^{16}$} &  &  &  &  & \textbf{2}...49\% & \textbf{2}...100\% & \textbf{2}...31\% & \textbf{2}...100\% & \textbf{2}...100\% \\
  & \textbf{3}...63\% & \textbf{3}...98\% & \textbf{3}...100\% & \textbf{3}...100 & \textbf{3}...51\% &  & \textbf{3}...69\% &   &  \\
  & \textbf{4}...37\% & \textbf{4}... 2\% &  &  &  &  &  &  &  \\
  \hline\multirow{3}{*}{$2^{20}$} &  &  &  &  &  & \textbf{2}...100\% &  & \textbf{2}...100\% & \textbf{2}...100\% \\
  &  & \textbf{3}...96\% & \textbf{3}...100\% & \textbf{3}...100\% & \textbf{3}...100\% &  & \textbf{3}...100\% &  &  \\
  & \textbf{4}...100\% & \textbf{4}... 4\% &  &  &  &  &  &  &  \\
  \hline\multirow{3}{*}{$2^{24}$} &  &  &  &  &  & \textbf{2}...100\% &  & \textbf{2}...100\% & \textbf{2}...100\% \\
  &  & \textbf{3}...37\% & \textbf{3}...100\% & \textbf{3}...100\% & \textbf{3}...100\% &  & \textbf{3}...100\% &  &  \\
  & \textbf{4}...100\% & \textbf{4}...63\% &  &  &  &  &  &  &  \\
  \hline
\end{tabular}
}
\label{tab:experiments}
\end{table}
 
We experimentally compare the performance of $\fd{d}$ with $\lt{d}$ and $\gd{d}$ in Table~\ref{tab:experiments}. Similar to the experimental results in \cite{V03}, we perform all 3 algorithms in different configurations of bins and $d$ values. Let $k$ be the maximum  number of probes allowed to be used by \fd{d} per ball. For each value of $d \in [2,4]$, we choose a corresponding value of $k$ such that the average number of probes required by each ball in $\fd{d}$ is at most $d$. For each configuration, we run each algorithm 100 times and note the percentage of times the maximum loaded bin had a particular value. It is of interest to note that $\fd{d}$, despite using on average less than $d$ probes per ball, appears to perform better than both $\gd{d}$ and $\fd{d}$ in terms of maximum load. 


\section{Conclusions and Future Work}
\label{sec:conc}
 
In this paper, we have introduced a novel algorithm called $\fd{d}$ for the well-studied load balancing problem.  This algorithm combines the benefits of two prominent algorithms, namely, $\gd{d}$ and $\lt{d}$. $\fd{d}$ generates a maximum load comparable to that of $\lt{d}$, while being as fully decentralized as $\gd{d}$. From another perspective, we observe that  $\fd{\log d}$  and $\gd{d}$ result in a comparable maximum load, while the number of probes used by $\fd{\log d}$ is exponentially smaller than that of $\gd{d}$.   In other words, we exhibit an algorithm that performs as well as an optimal algorithm, with significantly less computational requirements.  We believe that our work has opened up a new family of algorithms that could prove to be quite useful in a variety of contexts spanning both theory and practice. 

A number of questions arise out of our work. From a theoretical perspective, we are interested in developing a finer-grained analysis of the number of probes; experimental results suggest the number of probes used to place the $i^{th}$ ball depends on the congruence class of $i$ modulo $n$. From an applied perspective, we are interested in understanding how $\fd{d}$ would play out in real world load balancing scenarios like cloud computing, where the environment (i.e. the servers, their interconnections, etc.) and the workload (jobs, applications, users, etc.) are likely to be a lot more heterogeneous and dynamic.


\section*{Acknowledgements}
We are thankful to Anant Nag for useful discussions and developing a balls-in-bins library \cite{Nag14} that was helpful for our experiments. We are also grateful to Thomas Sauerwald for his helpful thoughts when he visited  Institute for Computational and Experimental Research in Mathematics (ICERM) at Brown University. Finally, John Augustine and Amanda Redlich are thankful to ICERM for having hosted them as part of a semester long program.

\bibliographystyle{siamplain}
\bibliography{references}

\end{document}